\documentclass[letterpaper, 10 pt, conference]{ieeeconf}

\IEEEoverridecommandlockouts

\usepackage{braket}
\usepackage{cite}
\usepackage{amsmath,amssymb,amsfonts,ntheorem}
\usepackage{algorithmic}
\usepackage{graphicx,pgfplots}
\usepackage{textcomp}
\usepackage{cite}
\usepackage{color}

\newtheorem{definition}{Definition}
\newtheorem{theorem}{Theorem}
\newtheorem{lemma}{Lemma}
\newtheorem{corollary}{Corollary}
\newtheorem{construction}{Construction}

\title{\LARGE \bf
Bounds and Constructions for\\ Multi-Symbol Duplication Error Correcting Codes
}

\author{Andreas Lenz, Niklas J\"unger and Antonia Wachter-Zeh%
\thanks{This work was supported by the Institute for Advanced Study (IAS), Technische Universit\"{a}t M\"{u}nchen (TUM), with funds from the German Excellence Initiative and the European Union's Seventh Framework Program (FP7) under grant agreement no.~291763}%
\thanks{Andreas Lenz, Niklas J\"unger and Antonia Wachter-Zeh are with the Institute of Communications Engineering, Technical University of Munich, 80333 Munich, Germany.}%
\thanks{E-Mail: {\tt andreas.lenz@mytum.de}, {\tt niklas.juenger@tum.de}, {\tt antonia.wachter-zeh@tum.de}}%
}

\begin{document}

\maketitle
\thispagestyle{empty}
\pagestyle{empty}

\begin{abstract}

In this paper, we study codes correcting $t$ duplications of $\ell$ consecutive symbols. These errors are known as tandem duplication errors, where a sequence of symbols is repeated and inserted directly after its original occurrence. Using sphere packing arguments, we derive non-asymptotic upper bounds on the cardinality of codes that correct such errors for any choice of parameters. Based on the fact that a code correcting insertions of $t$ zero-blocks can be used to correct $t$ tandem duplications, we construct codes for tandem duplication errors. We compare the cardinalities of these codes with their sphere packing upper bounds. Finally, we discuss the asymptotic behavior of the derived codes and bounds, which yields insights about the tandem duplication channel.

\end{abstract}

\section{Introduction}
Significant recent advances in creating, storing and reading deoxyribonucleic acid (DNA) molecules have paved the way for new storage technologies of digital data based on these macromolecules. According to the National Human Genome Research Institute (NHGRI), the sequencing costs per mega base pair (MBP) have dropped by a factor of $10^4$ to less than $\$ 0,1 / \mathrm{MBP}$ over the last 10 years \cite{Wet18}, which renders storage technologies based on DNA as a competitive candidate for long-term reliable and high-density archival storage. The challenge for reliable storage within such systems is the occurrence of new types of errors within the DNA molecules. Typical error types are point insertion, deletion and substitution errors, and also multi-nucleotide repetitions, such as tandem and palindromic duplications. The latter are errors, where a subword of the DNA is duplicated and inserted directly after the original string, and additionally reversed for the case of palindromic duplications. While substitution errors are well studied, little is known about the correction of insertion, deletion and duplication errors. An example for a tandem duplication of length~$3$ in a DNA sequence $GGCTAT$ is for example $GGCTA\underline{CTA}T$, where the underlined part highlights the duplication. Note that these duplication errors are special burst insertion errors, which form a type of errors, where a string of $\ell$ consecutive random symbols is inserted into a word. 

The design and analysis of codes correcting substitution errors and the underlying Hamming metric has been extensively studied since over 70 years. However, the knowledge about other error types, such as insertion and deletion errors is relatively limited. As such, there are only few codes that can correct more than a single insertion or deletion \cite{HF02,BGZ18}. Non-asymptotic upper bounds on the cardinalities of codes in the Levenshtein metric have only been found recently \cite{FVY15,KK13}. Duplication errors, which can be considered as a special kind of insertion errors have first been studied in the work of Levenshtein \cite{Lev65}, where codes correcting single symbol duplications and their asymptotic cardinality upper bounds have been derived for binary alphabets. More recently, in \cite{DA10} asymptotically optimal codes correcting duplication errors and prefixing methods to construct block codes have been found. An explicit construction with efficient encoding for multiple duplications has been presented in \cite{MV17}. These codes use the well-known relation of duplications and errors in the $\ell_1$-metric and are based on Lee-metric BCH codes \cite{RS94}. In \cite{JFSB16}, an optimal construction for the correction of an infinite number of fixed length $\ell$ duplications has been derived. Further, codes that correct duplications of variable lengths, with length up to $2$ or $3$ have been found. Both codes are constructed using results from the formal language theory and by choosing irreducible words with respect to tandem duplications. In~\cite{LWY17}, non-asymptotic upper bounds on codes correcting single tandem and palindromic duplications have been derived and in \cite{LWY18} several codes correcting single tandem and palindromic duplication errors have been presented.

In this paper, we extend these results and derive code constructions that correct $t$ tandem duplications, each of same length $\ell$ over arbitrary alphabets $\mathbb{Z}_p$ for any code length $n$. This generalizes the construction in \cite{DA10} to arbitrary alphabets and arbitrary duplication lengths and also generalizes \cite{LWY18} to arbitrary number of errors. In Construction \ref{con:lara}, we will use results from \cite{JFSB16} to construct a code based on tandem duplication roots. We compare the cardinalities of these constructions with upper bounds derived from sphere packing arguments. Note that recently, asymptotic bounds on the size of codes correcting a fixed number of $t$ tandem duplications, each of fixed length $\ell$ have been found \cite{KT18b}. In contrast, here we focus on non-asymptotic bounds, which are valid for any $n,t,\ell$.
\section{Preliminaries}

Let $\mathbb{Z}_p=\set{0,1,2,\dotsc,p-1}$ denote the finite ring with $p$ elements, where operations on elements of $\mathbb{Z}_p$ are done modulo $p$. 
The set of all finite words over $\mathbb{Z}_p$ is denoted by $\mathbb{Z}_p^*$. 
The set $\mathbb{Z}_p^n$ denotes all words of length $n$ over $\mathbb{Z}_p$. 
The vector $\mathbf{x}=(x_1x_2 \dotsc x_n) \in \mathbb{Z}_p^n$ represents a word of length $n$ with $x_i \in \mathbb{Z}_p$ denoting the $i$-th symbol.
The concatenation of two finite words $\mathbf{x},\mathbf{y} \in \mathbb{Z}_p^*$ is $(\mathbf{x}\mathbf{y}) \in \mathbb{Z}_p^{|\mathbf{x}|+|\mathbf{y}|}$, where $|\mathbf{x}|$ and $|\mathbf{y}|$ denotes the length of the words $\mathbf{x}$ and $\mathbf{y}$, respectively. 
With the term subword of a word $\mathbf{x} \in \mathbb{Z}_p^n$ we refer to a sequence consisting of consecutive letters $(x_i x_{i+1} \dotsc x_{i+j-1}) \in \mathbb{Z}_p^j$ from $\mathbf{x}$ starting at position $i$ with length $j$. For a function $f(\mathbf{x})$, where $\mathbf{x} \in \mathbb{Z}_p^*$ and a set $\mathcal{A} = \{\mathbf{a}_1, \dots, \mathbf{a}_{|\mathcal{A}|}\} \subseteq \mathbb{Z}_p^*$, we write $f(\mathcal{A}) = f(\mathbf{a}_1) \cup \dots, f(\mathbf{a}_{|\mathcal{A}|})$ as the union of all function evaluations. A tandem duplication is defined as follows.
\begin{definition}
	Given the word $\mathbf{x} \in \mathbb{Z}_p^*$, a tandem duplication of length $\ell$ at position $0 \leq i \leq |x|-\ell$ in $\mathbf{x}$ is
	\begin{equation*}
	\tau_{\ell}(\mathbf{x},i) = (\mathbf{u}\mathbf{v}\mathbf{v}\mathbf{w}),
	\end{equation*}
	where $\mathbf{x}=(\mathbf{u}\mathbf{v}\mathbf{w})$ with $|\mathbf{u}|=i$ and $\mathbf{v} \in \mathbb{Z}_p^{\ell}$ is the duplicated part in $\mathbf{x}$.
	
	\label{def:tandem_duplication}
\end{definition}
We refer to the common definition of an error ball around $\mathbf{x}$, as stated in the following definition.
\begin{definition}
	Given the word $\mathbf{x} \in \mathbb{Z}_p^*$, the $t$-tandem duplication ball around $\mathbf{x}$ is
	\begin{equation*}
	B_{t}^{\tau_\ell}(\mathbf{x}) = \Set{\mathbf{y} \in \mathbb{Z}_p^*|\mathbf{y}=\tau_\ell(\dotsc(\tau_\ell(\mathbf{x}, i_1)\dotsc), i_h), h \leq t},
	\end{equation*}
	i.e., the set of all words $\mathbf{y} \in \mathbb{Z}_p^*$ that can be reached by inserting at most $t$ tandem duplications of fixed length $\ell$ into the word~$\mathbf{x}$.
	
	\label{def:error_spheres}
\end{definition}

Similarly, we define the error spheres, as the set of words that can be reached with \emph{exactly} $t$ errors.

\begin{definition}
	Given the word $\mathbf{x} \in \mathbb{Z}_p^n$, the $t$-tandem duplication sphere around $\mathbf{x}$ is
	\begin{equation*}
	S_{t}^{\tau_\ell}(\mathbf{x}) = \Set{\mathbf{y} \in \mathbb{Z}_p^*|\mathbf{y}=\tau_\ell(\dotsc(\tau_\ell(\mathbf{x}, i_1)\dotsc), i_t)},
	\end{equation*}
	i.e., the set of all words $\mathbf{y} \in \mathbb{Z}_p^*$ that can be reached by inserting exactly $t$ tandem duplications of fixed length $\ell$ into~$\mathbf{x}$.
	\label{def:error_surface}
\end{definition}
Based on this,
a $t$-tandem-duplication-correcting code is defined as follows.
\begin{definition}
	The $p$-ary code $\mathcal{C} \subseteq \mathbb{Z}_p^n$ with codeword length $n$ is $t$-tandem-duplication-correcting with respect to duplication length $\ell$, if for all distinct codewords $\mathbf{c}_i,\mathbf{c}_j \in \mathcal{C}$
	\begin{equation*}
	B_t^{\tau_\ell}(\mathbf{c}_i) \cap B_{t}^{\tau_\ell}(\mathbf{c}_j) = \emptyset.
	\end{equation*}%
	\label{def:general_code_definition}
\end{definition}
\vspace{-.5cm}
Since the error spheres are subsets of the error balls, $S_{t}^{\tau_\ell}(\mathbf{c}_i) \subseteq B_{t}^{\tau_\ell}(\mathbf{c}_i)$ for all codewords $\mathbf{c}_i$, the error spheres $S_{t}^{\tau_\ell}(\mathbf{c}_i)$ for all $\mathbf{c}_i \in \mathcal{C}$ are disjoint as well. Note that it is possible to define tandem \textit{deletions} as the inverse operation of tandem duplications, which has been used in \cite[Cor. 1]{LWY18} to formulate tight bounds on duplication correcting codes. 

The following map has been found to be useful for constructing codes correcting tandem duplications \cite{JFSB16}.
\begin{definition}
	For any $\mathbf{x} \in \mathbb{Z}_p^n$, we define the map
	$$ \phi(\mathbf{x}) = (\mathbf{y},\mathbf{z}) \in \mathbb{Z}_p^\ell \times \mathbb{Z}_p^{n-\ell}, $$
	where
	\begin{align*}
		\mathbf{y} &= (x_1,x_2,\dots,x_\ell)\\
		\mathbf{z} &= (x_\ell-x_1, x_{\ell+1}-x_2,\dots,x_n-x_{n-\ell+1}).
	\end{align*}
\end{definition}
Note that this mapping is bijective, and therefore $\phi^{-1}$ is well defined. This mapping is useful since a tandem duplication is translated into a insertions of zero-blocks, when applied to $\phi$. Based on the mapping $\phi$, the $\ell$-duplication root, i.e. the word that is obtained by removing all length $\ell$ duplicates from $\mathbf{x}$ is defined as follows.
\begin{definition} \label{def:mu}
	For any word $\mathbf{x} \in \mathbb{Z}_p^n$ with $\phi(\mathbf{x}) = (\mathbf{y},0^{b_1} u_1 0^{b_1} u_2 \dots u_r 0^{b_{r+1}} )$, where $\mathbf{y} \in \mathbb{Z}_p^\ell$ and $u_i \in \mathbb{Z}_p \setminus \{0\}$, we define the $\ell$-duplication root of $\mathbf{x}$ by
	
	$$ \mu(\mathbf{x})=\left(0^{b_1\, (\bmod \ell)} u_1 0^{b_2 \,(\bmod \ell)} u_2 \dots u_r 0^{b_{r+1} \,(\bmod \ell)} \right).$$
\end{definition}
\section{Sphere Packing Bound}

In this section, we derive two sphere packing upper bounds for the size of a $t$-tandem duplication correcting code $\mathcal{C}$ with codeword length $n$. These upper bounds will make use of the fact that certain partitions of $\mathbb{Z}_p^n$ have the property that words of two distinct parts have non-overlapping error balls. It is therefore possible to partition every code accordingly and bound each of these partitions. This statement is summarized in the following theorem.
\begin{theorem}
	Let $\mathcal{C} \subseteq \mathbb{Z}_p^n$ be a $t$-tandem duplication correcting code and $\psi_1,\dotsc,\psi_k \subseteq \mathbb{Z}_p^n$ be a partition of $\mathbb{Z}_p^n$.
	Then
	\begin{equation*}
		|\mathcal{C}| \leq \sum_{i=1}^{k} \frac{|S_t^{\tau_\ell}(\psi_i)|}{\min\limits_{\mathbf{x} \in \psi_i } |S_t^{\tau_\ell}(\mathbf{x})|}.
	\end{equation*}
	
	\label{thm:sp_bound} 
\end{theorem}

\begin{proof}
	The space $\mathbb{Z}_p^n$ is divided into the $k$ disjoint parts $\psi_1,\dotsc,\psi_k \subseteq \mathbb{Z}_p^n$. 
	It follows from $\mathcal{C} \subseteq \mathbb{Z}_p^n$ that the code $\mathcal{C}$ is as well separated into $k$ disjoint subsets $\mathcal{C}_i=\mathcal{C}\cap \psi_i$ for $1 \leq i \leq k$, which implies $\mathcal{C}_i \subseteq \psi_i$.
	For the error spheres it directly follows $S_t^{\tau_\ell}(\mathcal{C}_i) \subseteq S_t^{\tau_\ell}(\psi_i)$, and for their cardinalities $|S_t^{\tau_\ell}(\mathcal{C}_i)| \leq |S_t^{\tau_\ell}(\psi_i)|$ for all $i$. 	
	Furthermore, $|S_t^{\tau_\ell}(\mathcal{C}_i)|$ is bounded from below by 
	\begin{equation*}
		|\mathcal{C}_i|\min\limits_{\mathbf{x} \in \psi_i } |S_t^{\tau_\ell}(\mathbf{x})| \leq |\mathcal{C}_i|\min\limits_{\mathbf{c} \in \mathcal{C}_i } |S_t^{\tau_\ell}(\mathbf{c})| \leq|S_t^{\tau_\ell}(\mathcal{C}_i)|,
		\label{ineq:SPproof2}
	\end{equation*}
	where we used the fact, that $S_t^{\tau_\ell}(\mathbf{c}_1) \cap S_t^{\tau_\ell}(\mathbf{c}_2) = \emptyset$ for $\mathbf{c}_1 \neq \mathbf{c}_2 \in \mathcal{C}_i$. The cardinality of the code $\mathcal{C}= \mathcal{C}_1 \cup \dotsc \cup \mathcal{C}_k$ over all disjoint partitions $\mathcal{C}_i$ is equal to the sum of all partial cardinalities $|\mathcal{C}_i|$ and we obtain the upper bound
	\begin{equation*}
		|\mathcal{C}| \leq \sum_{i=1}^{k} \frac{|S_t^{\tau_\ell}(\psi_i)|}{\min\limits_{\mathbf{x} \in \psi_i } |S_t^{\tau_\ell}(\mathbf{x})|}
	\end{equation*}
	stated in the theorem.
\end{proof}

\subsection{Upper Bound for Tandem Duplication Errors}

\begin{lemma}
	For duplication length $\ell$, the size of any $t$-tandem-duplication-correcting code $\mathcal{C} \subseteq \mathbb{Z}_p^n $ is upper bounded by
	\begin{equation*}
		|\mathcal{C}| \leq \frac{t!}{(n+(t-1)\ell)^t} \frac{p^{n+t(\ell+1)}}{(p-1)^t}.
	\end{equation*}
	
	\label{lemma:sphere_packing_bound_simple}
\end{lemma}
\begin{proof}
	We consider the partition $\psi_r = \Set{\mathbf{x} \in \mathbb{Z}_p^n : \phi(\mathbf{x}) = (\mathbf{y},\mathbf{z}), \mathrm{wt}_{\mathrm{H}}(\mathbf{z}) =r }$.	The Hamming weight of $\mathbf{z}$ can take values from $0$ to $n-\ell$ and therefore by Theorem \ref{thm:sp_bound} we obtain
	\begin{equation*}
		|\mathcal{C}| \leq \sum_{r=0}^{n-\ell} 		\frac{|S_t^{\tau_\ell}(\psi_r)|}{\min\limits_{\mathbf{x} \in \psi_r } |S_t^{\tau_\ell}(\mathbf{x})|}.
	\end{equation*}
	Since a tandem duplication in $\mathbf{x}$ corresponds to an insertion of a zero-block of length $\ell$ in $\phi(\mathbf{x})$, the number of all possible words after corrupting the word $\mathbf{x} \in \psi_r$ with $t$ tandem duplications is $\binom{r+t}{t}$, independently of $\mathbf{x}$. This is because, there are $r+1$ bins, where zero-blocks can be inserted. To compute $|S_t^{\tau_\ell}(\psi_r)|$, we start by counting all words $\mathbf{z'} \in \mathbb{Z}_p^{n+(t-1)\ell}$, that can be reached by inserting exactly $t$ blocks of $\ell$ consecutive zeros to any $\mathbf{z} \in \mathbb{Z}_p^{n-\ell}$ with Hamming weight $\mathrm{wt}_{\mathrm{H}}(\mathbf{z})=r$. This quantity is equal to the number of all $p$-ary words of length $n+(t-1)\ell$, Hamming weight $r$, that contain at least $t$ blocks of $\ell$ consecutive zeros and it is therefore upper bounded by the number of all $p$-ary words of length $n+(t-1)\ell$ and Hamming weight $r$. Since the first $\ell$ symbols of $\mathbf{x} \in \psi_r$ can be chosen arbitrarily and $S_{t}^{\tau_\ell}(\mathbf{x}) \cap S_{t}^{\tau_\ell}(\mathbf{y}) = \emptyset$ for $(x_1, \dots, x_\ell) \neq (y_1, \dots, y_\ell)$, i.e., the error spheres are distinct for different prefixes, we have
	\begin{equation*}
		|S_{\ell}^t(\psi_r)| \leq p^{\ell} \binom{n+(t-1)\ell}{r}(p-1)^r,
	\end{equation*}
	We therefore obtain for the upper bound
	\begin{align*}
		|\mathcal{C}| &\leq p^{\ell} \sum_{r=0}^{n-\ell} \frac{\binom{n+(t-1)\ell}{r}}{\binom{r+t}{t}}(p-1)^r \\
		&= \frac{p^{n+t(\ell+1)}(n+(t-1)\ell)!t!}{(p-1)^t (n+(t-1)\ell+t)!} \sum_{r=t}^{n-\ell+t} B^{n+(t-1)\ell+t}_\frac{p-1}{p}(r),
	\end{align*}
	where $B^M_{q}(k) = \binom{M}{k} q^k(1-q)^{M-k}$ is the probability mass function of the binomial distribution with $M$ trials and success probability $q$. The sum in the above equation is therefore upper bounded by $1$ and we obtain
	\begin{align*}
		|\mathcal{C}| &\leq t! \frac{(n+(t-1)\ell)!}{(n+(t-1)\ell+t)!} \frac{p^{n+t(\ell+1)}}{(p-1)^t} \\
		&\leq \frac{t!}{(n+(t-1)\ell)^t} \frac{p^{n+t(\ell+1)}}{(p-1)^t},
	\end{align*}
	where we used $\frac{a!}{(a+b)!} \leq \frac{1}{a^b}$ for any $a,b \in \mathbb{N}$.
\end{proof}
From Lemma \ref{lemma:sphere_packing_bound_simple}, we can deduce an asymptotic bound on the cardinality of codes correcting $t$ tandem duplications.
\begin{corollary}
	For any $\ell$, the size of any $t$-tandem-duplication-correcting code $\mathcal{C} \subseteq \mathbb{Z}_p^n $ satisfies asymptotically for $n \rightarrow \infty$
	$$|\mathcal{C}| \lesssim \frac{t! p^{n+t(\ell+1)}}{(n(p-1))^t},$$
	where $f(n) \lesssim g(n)$ means that $\lim_{n\rightarrow \infty} \frac{f(n)}{g(n)} \leq 1$.
\end{corollary}
The upper bound in Lemma \ref{lemma:sphere_packing_bound_simple} for the code cardinality directly implies a lower bound on the redundancy of any code $\mathcal{C}\subseteq\mathbb{Z}_p^n$ that corrects $t$ tandem duplication of length $\ell$.
\begin{lemma} \label{lem:c1c}
	For duplication length $\ell$, the redundancy of any $t$-tandem-duplication-correcting code $\mathcal{C} \subseteq \mathbb{Z}_p^n $ is at least
	\begin{equation*}
	n-\log_p |\mathcal{C}| \geq t\log_p n - t\log_p (t(p-1)) - t(\ell+1).
	\end{equation*}
\end{lemma}
\begin{proof}
	To obtain the lower bound of the redundancy, we use the upper bound onto the codebook size, stated in Lemma \ref{lemma:sphere_packing_bound_simple} and get
	\begin{align*}
	n-\log_p |\mathcal{C}| &\geq n - \log_p \left( \frac{t!}{(n+(t-1)\ell)^t} \frac{p^{n+t(\ell+1)}}{(p-1)^t} \right) \\
	&= t \log_p \left(n+(t-1)\ell\right)+t \log_p (p-1)\\ 
	&\qquad -\log_p ( t!)-t(\ell+1) \\
	&\geq t\log_p n - t\log_p (t(p-1)) - t(\ell+1).
	\end{align*}%
\end{proof}
Lemma \ref{lem:c1c} indicates that the minimum required redundancy of a $t$-tandem-duplication-correcting code $\mathcal{C} \subseteq \mathbb{Z}_p^n$ decreases with increasing duplication length $\ell$. It can further be seen that similar to the case of conventional errors, the redundancy scales as $t \log_pn$ for growing code lengths $n$.

\subsection{Stronger Sphere Packing Bound}

In this section we derive a sphere packing bound based on the fact that words with distinct duplication roots have non-overlapping error balls for any $t$, which has been shown in~\cite{JFSB16}. Here we use the partition $\psi_i = \Set{\mathbf{x} \in \mathbb{Z}_p^n : \mu(\mathbf{x})=\mathbf{r}_i}$ over all duplication roots $\mathbf{r}_i$ to divide $\mathbb{Z}_p^n$ into disjoint subsets, where the words within a subset belong to the same root. 
Similar elaborations as for the previous sphere packing bound result in the following bound.
\begin{lemma} \label{lemma:sp:complex}
	The cardinality of any $t$ tandem duplication correcting code over $\mathbb{Z}_p^n$ is upper bounded by
	\begin{equation*}
	|\mathcal{C}| \leq p^{\ell} \sum_{w=0}^{\lfloor \frac{n}{\ell} \rfloor -1} \sum_{r=0}^{n-(w+1)\ell} N_{p,\ell}(n-(w+1)\ell,r) \frac{\binom{r+w+t}{w+t}}{\binom{r+t}{t}},
	\end{equation*}
	where $N_{p,\ell}(n,r)$ is the number of words in $\mathbb{Z}_p^n$ with Hamming weight~$r$, which have less than $\ell$ consecutive zeros. 
\end{lemma}

An explicit expression for $N_{p,\ell}(n,r)$ can be found in \cite{Kur11,LWY17}. Compared to the sphere packing bound in Lemma \ref{lemma:sphere_packing_bound_simple} this upper bound has a very complex expression, but in Section \ref{sec:con}, we perform numerical evaluations which show that the bound is stronger than the first one.
\section{Code Constructions}  \label{sec:con}
In this section, we propose two constructions that are able to correct $t$ tandem duplications, each of fixed length $\ell$. We mainly concentrate on codes correcting insertions of zero-blocks,  due to the equivalence of these errors to tandem duplications under the map $\phi$, as seen in the previous section. First, we derive a construction that is able to correct $t$ insertions of zero-blocks, each of length $\ell$.
\subsection{Construction based on \cite{DA10}}
\begin{definition} For some integers $r,n$ and $\xi$ with $r\leq n$ and additionally a vector $\mathbf{a} \in \{0,1,\dots,\xi-1\}^t$, we define
	\begin{align*}
		\mathcal{C}_{p,t,\ell}(n,r, \mathbf{a}, \xi) &= \Bigg\{  \mathbf{z} \in \mathbb{Z}_p^{n} : \mathrm{wt}_\mathrm{H}(\mathbf{z}) = r, \\ & \sum_{i=1}^{r+1} i^q \left\lfloor \frac{b_i}{\ell} \right\rfloor = a_q \;(\bmod \,\xi) \; \forall \; 1\leq q \leq t \Bigg\},
	\end{align*} 
	where $b_i$ denotes the number of zeros in $\mathbf{z}$, which have exactly $i-1$ non-zero symbols to the left, i.e. $\mathbf{z} = (0^{b_1} u_1 0^{b_1} u_2 \dots u_r 0^{b_{r+1}} )$ with non-zero letters $u_i \in \mathbb{Z}_p \setminus \{0\}$.
\end{definition}
By this definition, $\mathcal{C}_{p,t,\ell}(n,r,\mathbf{a},\xi) \subseteq \mathbb{Z}_p^n$ contains words of Hamming weight $r$, which satisfy the $t$ checksum constraints. When $\xi$ is chosen carefully, we will show that this code corrects $t$ insertions of zero-blocks, each of length $\ell$. 
\begin{lemma} \label{lemma:lara:correcting}
	The code $\mathcal{C}_{p,t,\ell}(n,r,\mathbf{a},\xi)$ is $t$ zero-block (length $\ell$) insertion correcting for any $\mathbf{a} \in \{0,1,\dots, \xi-1 \}^t$, if $\xi$ is a prime number satisfying $\xi > \max(t,r)$.
\end{lemma}
\begin{proof}
	We will use a similar proof technique as in \cite{DA10} to show that $\mathcal{C}_{p,t,\ell}(n,r,\mathbf{a},\xi)$ is $t$ zero-block insertions correcting. Denote by $i_1, \dots, i_t$ the positions of the insertions in $\mathbf{z}$, i.e. the $j$-th insertion occurred in the $i_j$-th run of zeros. Since $\mathrm{wt_\mathrm{H}}(\mathbf{z})=r$, it follows that $1\leq i_j \leq r+1$ for all $1\leq j \leq t$. Computing the checksum deficiencies yields
	\begin{align*}
		\tilde{s}_q &= \tilde{a}_q-a_q= \sum_{i=1}^{r+1} i^q \left\lfloor \frac{\tilde{b}_i}{\ell} \right\rfloor - a_q \,\, &&(\bmod \, \xi) \\&= i_1^q + i_2^q + \dots + i_t^q &&(\bmod \,\xi),
	\end{align*}
	for all $1 \leq q \leq t$. In \cite{DA10}, it has been shown that this system of equations has a unique solution if $\xi$ is a prime number satisfying $\xi > \max (t,r)$. This allows to uniquely identify the positions of the inserted zero-blocks and thus allows to recover the original word.
\end{proof}
\begin{construction}\label{con:1}
	Let
	$$ \mathcal{C}^1 = \{ \phi^{-1}(\mathbf{y}, \mathbf{z}): \mathbf{y} \in \mathbb{Z}_p^\ell, \mathbf{z} \in \mathcal{C}_{p,t,\ell}(n-\ell,r,\mathbf{a}_r^*,\xi_r) \}, $$
	where $r = \mathrm{wt_H}(\mathbf{z})$. Further, $\xi_r$ is a prime number satisfying $\max(t,r)<\xi_r \leq 2 \max(t,r)$ and $$\mathbf{a}_r^* = \underset{\mathbf{a}\in \{0,1,\dots,\xi-1\}^t}{\arg \max}\, |\mathcal{C}_{p,t,\ell}(n-\ell,r,\mathbf{a},\xi_r)|.$$
\end{construction}
Since $\mathbf{z}$ is chosen to be contained in $\mathcal{C}_{p,t,\ell}(n-\ell,r,\mathbf{a}_r^*, \xi_r)$, the code $\mathcal{C}^1$ is able to correct $t$ tandem duplications of length~$\ell$. The following lemma on the cardinality of Construction \ref{con:1} directly follows from the pigeonhole principle.
\begin{lemma} \label{lem:card:c1}
	The cardinality of $\mathcal{C}^1$ satisfies
	$$ |\mathcal{C}^1| \geq p^\ell \sum_{r=0}^{n-\ell} \frac{\binom{n-\ell}{r}(p-1)^r}{(2\max(t,r))^t}. $$
\end{lemma}
\begin{proof}
	We begin by observing that for any $\xi_r \in \mathbb{N}$, each word $\mathbf{z} \in \mathbb{Z}_p^\ell$ with $\mathrm{wt_H}(\mathbf{z}) = r$ is contained in $\mathcal{C}_{p,t,\ell}(n-\ell,r,\mathbf{a}_r, \xi_r)$ for some $\mathbf{a}_r \in \{0,1,\dots,\xi_r-1\}^t$. Therefore,
	$$ \left| \bigcup_{\mathbf{a}_r } \mathcal{C}_{p,t,\ell}(n-\ell,r,\mathbf{a}_r, \xi_r) \right| = \binom{n-\ell}{r}(p-1)^r, $$
	which implies that there exists an $\mathbf{a}_r^*$ such that
	$$ |\mathcal{C}_{p,t,\ell}(n-\ell,r,\mathbf{a}_r^*,\xi_r)| \geq \frac{\binom{n-\ell}{r}(p-1)^r}{\xi_r^t}, $$
	since there are $\xi_r^t$ distinct vectors $\mathbf{a}_r$. By Lemma \ref{lemma:lara:correcting}, $\xi_r$ has to be a prime that satisfies $\xi_r > \max(t,r)$. Using Bertrand's postulate, there exists a prime $\max(t,r) < \xi_r \leq 2\max(t,r) $. Choosing $\mathbf{y} \in \mathbb{Z}_p^\ell$ arbitrarily yields the lemma.
\end{proof}
Further, the following lemma can be shown for the asymptotic behavior of the cardinality of $\mathcal{C}^1$.
\begin{lemma}
	The cardinality of $\mathcal{C}^1$ satisfies asymptotically
	$$ |\mathcal{C}^1| \gtrsim \frac{p^{n+t}}{(n(p-1))^t}. $$
\end{lemma}
\begin{proof}
	The cardinality of Construction \ref{con:1} is lower bounded by
	$$ |\mathcal{C}^1| \geq p^\ell \sum_{r=0}^{n-\ell} \frac{\binom{n-\ell}{r}(p-1)^r}{\xi_r^t}, $$
	where $\xi_r > \max(t,r)$ are primes.
	Using the result about the existence of prime numbers from \cite[Section 4.1]{DA10}, yields
	\begin{align*}
		 |\mathcal{C}^1| & \gtrsim p^\ell \sum_{r=t}^{n-\ell} \frac{\binom{n-\ell}{r}(p-1)^r}{r^t} \\
		 & \geq \frac{p^\ell (n-\ell)!}{(n-\ell+t)!} \sum_{r=t}^{n-\ell} \binom{n-\ell+t}{r}(p-1)^r \gtrsim \frac{p^{n+t}}{n^t (p-1)^t},
	\end{align*}
	which proves the lemma.
\end{proof}
\subsection{Construction using Duplication Roots}
Next we introduce a construction based on duplication roots. A duplication root (duplication length $\ell$) is a word, that does not contain any duplications of length $\ell$. For more detailed information about duplication roots, we refer the reader to \cite{JFSB16}. We introduce a function, which maps a tandem duplication of length $\ell$ to an insertion of a single zero.
\begin{definition}
	For any word $\mathbf{x} \in \mathbb{Z}_p^n$ with $\phi(\mathbf{x}) = (0^{b_1} u_1 0^{b_1} u_2 \dots u_r 0^{b_{r+1}} ) \in \mathbb{Z}_p^n$ and non-zero letters $u_i \in \mathbb{Z}_p \setminus\{0\}$, we define the injective map $$\mathcal{T}(\mathbf{x})=(y(\mathbf{x}), \mu(\mathbf{x}),\pi(\mathbf{x}))$$ with
	\begin{align*}
		y(\mathbf{x})&=(x_1\dots x_\ell) \\
		\pi (\mathbf{x})&=\left( 0^{ \lfloor \frac{b_1}{\ell} \rfloor } 1 \dots 1 0^{ \lfloor \frac{b_{r+1}}{\ell} \rfloor }  \right),
	\end{align*}
	and $\mu(\mathbf{x})$ as in Definition \ref{def:mu}.
\end{definition}
It can be shown that a tandem duplication in $\mathbf{x}$ directly translates into an insertion of a single zero in $\pi(\mathbf{x})$ of $\mathcal{T}(\mathbf{x})$ and leaves $y(\mathbf{x})$ and $\mu(\mathbf{x})$ unchanged.
For given $\mathbf{y} \in \mathbb{Z}_p^\ell$ and $\mu \in \mathbb{Z}_p^{n-\ell}$, we consider the set $\rho(\mathbf{y},\mu)=\{ \mathbf{x} \in \mathbb{Z}_p^n : y(\mathbf{x}) = \mathbf{y} \wedge \mu(\mathbf{x}) = \mu \}$.
Then the code $\mathcal{C} \subseteq \mathbb{Z}_p^n$ is $t$-tandem-duplication-correcting if for all $\mathbf{y}, \mu$ the binary set $\pi(\mathcal{C} \cap \rho(\mathbf{y},\mu))$ is $t$-single-zero-insertion-correcting.
For the set $\pi(\mathcal{C} \cap \rho(\mathbf{y},\mu))$ we use a single zero insertion correcting set
proposed in \cite{DA10} to obtain the following construction.

\begin{construction} \label{con:lara}
	We define the code of length $n$ over $\mathbb{Z}_p^n$, for some duplication length $\ell \in \mathbb{N}$ as
	\begin{equation*}
	\mathcal{C}^2 = \bigcup_{w=0}^{\left\lfloor \frac{n}{\ell} \right\rfloor-1} \Set{\mathbf{x} \in \mathbb{Z}_p^n | 
		\begin{aligned}
		&\mathrm{wt_H}(\mu(\mathbf{x})) = r \\ 
		& \mu(\mathbf{x}) \in \mathbb{Z}_p^{n-(w+1)\ell} \\
		& \pi(\mathbf{x}) \in \mathcal{C}_{2,t,1}(r+w,r,\mathbf{a}_r^*,\xi_r)
		\end{aligned} }.
	\end{equation*} 
	\label{def:code_b}
\end{construction}
Since a tandem duplication corresponds to a single insertion of a zero in $\pi(\mathbf{x})$, the code $\mathcal{C}^2$ is also $t$-tandem-duplication-correcting. Its cardinality can be derived in a similar fashion as for Construction \ref{con:1} and we obtain the following lemma.
\begin{lemma}
	The code $\mathcal{C}^2$ has a cardinality at least
	\begin{equation*}
	|\mathcal{C}^2| \geq  p^{\ell} \sum_{w=0}^{\left\lfloor \frac{n}{\ell} \right\rfloor-1} \sum_{r=0}^{n-(w+1)\ell}  \frac{N_{p,\ell}(n-(w+1)\ell,r)\binom{r+w}{w}}{(2\max(r,t)+1)^t},
	\end{equation*}
	where $N_{p,\ell}(n,r)$ is the number of words in $\mathbb{Z}_p^n$ with Hamming weight~$r$, which have less than $\ell$ consecutive zeros. 
	\label{corollary:codebook_size}
\end{lemma}
Note that for $t \rightarrow \infty$, Construction \ref{con:lara} coincides with \cite[Const. A]{JFSB16}. Figs. \ref{fig:t1} and \ref{fig:t3} show the redundancy of Construction \ref{con:1} and \ref{con:lara} for several parameters and their lower bounds from Lemma \ref{lemma:sphere_packing_bound_simple} and \ref{lemma:sp:complex}. We further show the bound from \cite[Cor. 1]{LWY18}, gives a tight bound for the single error case, $t=1$. Note that \cite[Const. 1]{LWY18} yields similar redundancies as the proposed Construction \ref{con:1} for $t=1$. These results indicate that the cardinality of Construction \ref{def:code_b} is larger than the cardinality of Construction \ref{con:1}. This is because Construction \ref{def:code_b} makes efficient use of the fact that words with distinct duplication roots can never have overlapping error balls and therefore achieves larger cardinalities. Interestingly, the plots further suggest that the cardinality of Construction \ref{def:code_b} is close to optimal for some choices of parameters.

\begin{figure}[thpb]
	\begin{tikzpicture}
		\begin{axis}[
			width=8.25cm,
			height=6.5cm,
			xlabel=$n$,
			ylabel=$n-\log_p|\mathcal{C}|$,
			grid=both,
			xmin=5, xmax=60,
			ymin=-0.5,
			ymax=9,
			legend pos=north west,
			legend cell align={left},
			legend columns=3,
			transpose legend,
			]
			\addplot[color=black,solid] table [col sep=comma] {sp_simple.csv};
			\addlegendentry{SP, Lemma \ref{lemma:sphere_packing_bound_simple}};
			\addplot[color=black,dashed] table [col sep=comma] {sp_complex.csv};
			\addlegendentry{SP, Lemma \ref{lemma:sp:complex}};
			\addplot[color=black,dashdotted] table [col sep=comma] {sp_dcc_l_3.csv};
			\addlegendentry{SP, \cite[Cor. 1]{LWY18}};
			\addplot[color=black,mark=*] table [col sep=comma] {con1.csv};
			\addlegendentry{Const. \ref{con:1}};
			\addplot[color=black,mark=o] table [col sep=comma] {con2.csv};
			\addlegendentry{Const. \ref{con:lara}};
		\end{axis}
	\end{tikzpicture}
	\vspace{-.2cm}
	\caption{Redundancy of constructions and bounds $(p=2,t=1,\ell=3)$}
	\label{fig:t1}
\end{figure}
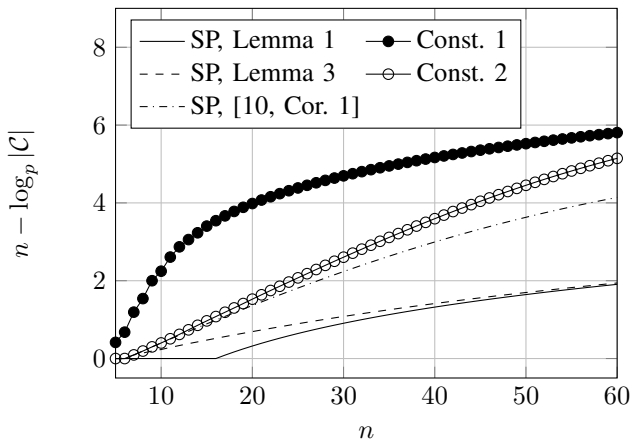

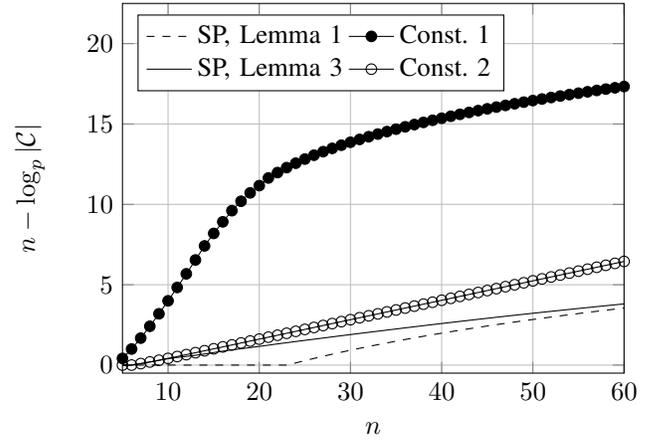
\begin{figure}[thpb]
	\begin{tikzpicture}
	\begin{axis}[
	width=8.25cm,
	height=6.5cm,
	xlabel=$n$,
	ylabel=$n-\log_p|\mathcal{C}|$,
	grid=both,
	xmin=5, xmax=60,
	ytick={0,5,10,15,20},
	ymin=-0.5,
	ymax=22.5,
	legend pos=north west,
	legend cell align={left},
	legend columns=2,
	transpose legend,
	]
	\addplot[color=black,dashed] table [col sep=comma] {sp_simple_t3.csv};
	\addlegendentry{SP, Lemma \ref{lemma:sphere_packing_bound_simple}};
	\addplot[color=black,solid] table [col sep=comma] {sp_complex_t3.csv};
	\addlegendentry{SP, Lemma \ref{lemma:sp:complex}};
	\addplot[color=black,mark=*] table [col sep=comma] {con1_t3.csv};
	\addlegendentry{Const. \ref{con:1}};
	\addplot[color=black,mark=o] table [col sep=comma] {con2_t3.csv};
	\addlegendentry{Const. \ref{con:lara}};
	\end{axis}
	\end{tikzpicture}
	\vspace{-.25cm}
	\caption{Redundancy of constructions and bounds $(p=2,t=3,\ell=3)$}
	\vspace{-.55cm}
	\label{fig:t3}
\end{figure}

\section{Conclusion}

In this paper, we have derived non-asymptotic upper bounds on the size of codes correcting $t$ tandem duplications, each of length $\ell$. We have further found two constructions can correct $t$ tandem duplications, each of length $\ell$. Our results indicate that correcting $t$ duplications requires less redundancy with increasing duplication length $\ell$, which coincides with the result from \cite{JFSB16}, where this property has been shown for correcting an arbitrary number of tandem duplications.
 
\vspace{-.135cm}

\bibliography{IEEEabrv,ref.bib}
\bibliographystyle{IEEEtranS}

\end{document}